\documentclass[11pt]{amsart}

\usepackage{mathtools}
\usepackage{amsmath}
\usepackage{amssymb, stmaryrd, bm}
\usepackage{amsthm}
\usepackage{graphicx}
\usepackage{tikz}
\usepackage{xcolor}
\usepackage{enumitem}
\usepackage{algorithm}

\usepackage[foot]{amsaddr}
\usepackage[pdftex,colorlinks,backref=page,citecolor=blue]{hyperref}
\mathtoolsset{showonlyrefs}

\allowdisplaybreaks

\usepackage[margin=1.2in]{geometry}
\usepackage{bookmark}

\def\E{\mathbb{E}}

\def\E{\mathbb{E}}

\def\P{\mathbb{P}}

\def\eps{\varepsilon}

\def\cC{\mathcal {C}}
\def\cD{\mathcal{D}}
\def\cE{\mathcal {E}}

\def\cG{\mathcal {G}}

\def\cL{\mathcal{L}}

\def\cT{\mathcal {T}}
\def\cI{\mathcal {I}}

\def\Ec{E_{\mathrm{cr}}}

\def\1{\mathbf{1}}

\def\lam {\lambda}

\def\tce{t_c + \eps}
\def\tce2{t_c + \frac{\eps}{2}}
\def\ER{Erd\H{o}s-R\'{e}nyi }

\newtheorem*{theorem*}{Theorem}
\newtheorem{theorem}{Theorem}
\numberwithin{theorem}{section}
\newtheorem{lemma}[theorem]{Lemma}
\newtheorem{cor}[theorem]{Corollary}
\newtheorem{defn}[theorem]{Definition}
\newtheorem*{defn*}{Definition}
\newtheorem{prop}[theorem]{Proposition}
\newtheorem*{prop*}{Proposition}

\newtheorem*{conj*}{Conjecture}

\newtheorem{question}{Question}

\newtheorem*{fact*}{Fact}

\numberwithin{equation}{section}

\begin{document}
\title{Sampling and counting triangle-free graphs \\ near the critical density}

\author{Matthew Jenssen}
\author{Will Perkins}
\author{Aditya Potukuchi}
\author{Michael Simkin}

\address{King's College London, Department of Mathematics}
\email{matthew.jenssen@kcl.ac.uk}

\address{Georgia Institute of Technology, School of Computer Science}
\email{math@willperkins.org}

\address{York University, Department of Electrical Engineering and Computer Science}
\email{apotu@yorku.ca}

\address{Massachusetts Institute of Technology, Department of Mathematics}
\email{msimkin@mit.edu}

\date{October 30, 2024.}

\begin{abstract}
We study the following combinatorial  counting and sampling  problems: can we efficiently sample from the Erd\H{o}s-R\'{e}nyi random graph $G(n,p)$ conditioned on triangle-freeness?  Can we efficiently approximate  the probability that $G(n,p)$ is triangle-free? These are prototypical instances of  forbidden substructure problems ubiquitous in combinatorics.  The algorithmic questions are instances of approximate counting and sampling for a hypergraph hard-core model. 

Estimating the probability that $G(n,p)$ has no triangles is a fundamental question in probabilistic combinatorics and one that has led to the development of many important tools in the field.  Through the work of several authors, the asymptotics of the logarithm of this probability are known if $p =o( n^{-1/2})$ or if $p =\omega( n^{-1/2})$.  The regime $p = \Theta(n^{-1/2})$ is more mysterious, as this range witnesses a dramatic change in the the typical structural properties of $G(n,p)$ conditioned on triangle-freeness.  As we show, this change in structure has a profound impact on the performance of sampling algorithms.

We give two different efficient sampling algorithms for triangle-free graphs (and complementary algorithms to approximate the triangle-freeness large deviation probability), one that is efficient when $p < c/\sqrt{n}$ and one that is efficient when $p > C/\sqrt{n}$ for  constants $c, C>0$.  The latter algorithm involves a new approach for dealing with large defects in the setting of sampling from low-temperature spin models.    

An extended abstract of this paper appeared at FOCS 2024~\cite{jenssen2024sampling}.
\end{abstract}

\maketitle

\section{Introduction}
\label{secIntro}

Let $\cT = \cT(n)$ denote the set of (labeled) triangle-free graphs on $n$ vertices. Understanding properties of this set is a central topic in extremal and probabilistic combinatorics. Mantel's Theorem~\cite{mantel1907problem}, one of the earliest results in extremal graph theory, solves an optimization problem: which $G \in \cT(n)$ has the most edges?  Erd\H{o}s, Kleitman, and Rothschild~\cite{erdosasymptotic} gave an asymptotic formula for $|\cT(n)|$ and proved an important structural property of the uniform distribution on $\cT(n)$: almost all triangle-free graphs are bipartite.  More generally, one can study $\mu_p(\cT)$, the probability that the \ER random graph $G(n,p)$ is triangle-free, a question considered already in the papers of Erd\H{o}s and R\'{e}nyi initiating the study of random graphs~\cite{erdos1960evolution}.

Understanding such non-existence probabilities and  more general large deviation probabilities for subgraphs in random graphs often involves deep mathematical tools from combinatorics and probability theory, including regularity lemmas, nonlinear large deviations, container methods, and entropic methods (e.g.~\cite{chatterjee2011large,chatterjee2017large,balogh2015independent,saxton2015hypergraph,chatterjee2016nonlinear,eldan2018gaussian,cook2020large,kozma2023lower}). These techniques often connect large deviation probabilities to structural information about graphs drawn from  the appropriate  conditional distribution.

Here we will be interested in the connection between structural properties of the conditional distribution, algorithms to sample from the conditional distribution, and algorithmic approximations of the large-deviation probability.

In our case, the distribution of interest is that of $G(n,p)$ conditioned on triangle-freeness. We denote this distribution by $\mu_{\cT,p}$. The    result of Erd\H{o}s, Kleitman, and Rothschild  not only gives an asymptotic formula for $\mu_{1/2}(\cT)$, but also gives an efficient algorithm to  sample approximately from $\mu_{\cT,1/2}$: sample a nearly balanced partition of a set of $n$ vertices and include each crossing edge independently with probability $1/2$.  This gives a  random bipartite graph on $n$ vertices, whose distribution is very close to $\mu_{\cT,1/2}$.   Osthus, Pr{\"o}mel, and Taraz~\cite{osthus2003densities} later showed that an analogous result holds for much smaller $p$: for any fixed $\eps>0$, if $p \ge(1+\eps) \sqrt{3 \log n}/ \sqrt{n}$ then $G \sim \mu_{\cT,p}$ is bipartite whp. In the same way as above, this immediately yields an efficient approximate sampling algorithm.

What happens for smaller $p$?  T.\ \L uczak~\cite{luczak2000triangle} had previously showed that if $p \ge C/\sqrt{n}$ for large $C$, then whp $G$ has a cut containing most of the edges; that is, it is \emph{approximately} bipartite.  On the other hand, when $p =o(n^{-1/2})$, no such large cut or global structure exists, and instead behavior is captured by the `Poisson paradigm' and Janson's Inequality~\cite{janson1987uczak}.  

Here we will be interested in the regime $p= \Theta(n^{-1/2})$, in which this global structure emerges. Beyond \L uczak's result and consequences of Janson's Inequality, very little is known in this regime about either the properties of $\mu_{\cT,p}$, the value of $\mu_{p}(\cT)$, or the existence of efficient sampling algorithms.  In this paper we will show that these three questions are closely related.

In different settings, very similar questions about approximation of probabilities and understanding conditional probability measures are studied in computer science in the area of approximate counting and sampling.  Here a canonical object is the partition function of the hard-core model or the generating function for independent sets of a graph or hypergraph $G$:
\begin{align}
\label{eqHCzdef}
    Z_G(\lam) = \sum_{I \in \cI(G)} \lam^{|I|} \,,
\end{align} 
where $\cI(G)$ is the set of independent sets of $G$.
The related sampling question is about sampling from the hard-core measure, the probability measure on $\cI(G)$ defined by
\[ \mu_{G,\lam}(I) = \frac{\lam^{|I|}}{Z_G(\lam)} \,.\]
The complexity of approximate sampling and counting for the hard-core model on graphs of maximum degree $\Delta$ is now very well understood: when $\lam$ is below a threshold $\lam_c(\Delta) \approx \frac{e}{\Delta}$, there are efficient sampling and counting algorithms~\cite{weitz2006counting,anari2021spectral,chen2021optimal,anari2022entropic,chen2022localization}, and when $\lam$ is above this threshold no such algorithms exist unless NP$=$RP~\cite{Sly10,SS14,galanis2016inapproximability}.  The critical value $\lam_c$ marks a phase transition of the hard-core model on the infinite $\Delta$-regular tree. The phase transition behavior is also reflected in finite graphs: when $\lam<\lam_c$, the hard-core model on any graph of maximum degree $\Delta$ exhibits several `nice' behaviors including decay of correlations and rapid convergence of local dynamics (Markov chains), while there exist graphs for which these properties fail for $\lam > \lam_c$ (in particular, for the random $\Delta$-regular bipartite graph~\cite{mossel2009hardness}).  

The probability that $G(n,p)$ is triangle-free and its conditional distribution given this event can be written in terms of a hard-core partition function and measure. Given a graph $G$ on vertex set $[n]$, we will often identify $G$ with its edge set and let $|G|$ denote the number of edges in $G$. Let $\mu_p$ denote the distribution of the \ER random graph $G(n,p)$, that is, $\mu_p(G)=p^{|G|}(1-p)^{\binom{n}{2}-|G|}$.   Let
\begin{align*}
Z(\lam) = \sum_{G \in \cT} \lam ^{|G|}\, .
\end{align*}
With $\lam = \frac{p}{1-p}$, we have the  identities
\begin{align}\label{eqPviaZ}
\mu_p(\cT)= \frac{Z(\lam)}{(1+\lam)^{\binom{n}{2}}}\, ,
\end{align}
and
\[
\mu_{\cT,p}(G)=\frac{\lam^{|G|}}{Z(\lam)},\, \quad G \in \cT \, .
\]

From this perspective, $\mu_{\cT,p}$ may be viewed as a hypergraph hard-core measure. Let $H=(V,E)$ be the hypergraph with $V=\binom{[n]}{2}$ where $\{x,y,z\}\in E$ if and only if $x,y,z$ form a triangle (when viewed as edges of a graph on vertex set $[n]$). Then $\mu_{\cT,p}=\mu_{H,\lam}$ where $\lam=p/(1-p)$.

The algorithmic and probabilistic behavior of hypergraph hard-core models on bounded-degree, uniform hypergraphs is more complex and less well understood than that of graphs.  While for graphs $\lam_c(\Delta)$ marks a computational threshold, a phase transition on a particular infinite graph, a threshold for fast mixing of local dynamics, and a threshold for a strong form of correlation decay (strong spatial mixing), in hypergraphs these thresholds need not coincide (see, e.g., the results and discussion in~\cite{hermon2019rapid,bezakova2019approximation,galvin2024zeroes,bencs2023optimal}).

Our main results combine these points of view and give some understanding of triangle-free graphs in this critical regime $p = \Theta(n^{-1/2})$.  We will give efficient approximate counting and sampling algorithms when $p \le c n^{-1/2}$ and $p \ge Cn^{-1/2}$ for constants $c,C>0$.  In the paper~\cite{jenssen2024lower}, we use algorithmic ideas to prove an asymptotic formula for $\log \mu_p(\cT)$ when $p \le c n^{-1/2}$ (and treat more general lower-tail problems for triangles in the critical regime).

To state our algorithmic results we need a couple of definitions.  The \textit{Glauber dynamics} is a Markov chain defined by a probability distribution $\mu$ on vectors (in our case indicator functions of edges of a graph) that  proceeds by choosing a random coordinate and updating this coordinate according to the conditional distribution of $\mu$ given all other coordinates.  The mixing time of a Markov chain is the number of steps required to guarantee convergence to the stationary distribution $\mu$ (to within  $1/4$ total variation distance) from a worst-case starting point.  See Section~\ref{secPrelim} for full definitions. 
We say that $\hat Z>0$ is an $\eps$-relative approximation to $Z>0$ if $e^{-\eps}\hat Z\leq Z \leq e^{\eps}\hat Z$.  
\begin{theorem}
\label{thmLowMixing}
Let $c < 1/\sqrt{2}$ be fixed and suppose $p \le cn^{-1/2}$. The Glauber dynamics for sampling from $\mu_{\cT,p}$ has mixing time $O(n^2 \log n)$.

In particular, this gives randomized algorithms running in time polynomial in $n$ and $1/\eps$ that
 \begin{enumerate}
        \item Output $M$  so that with probability at least $2/3$, $M$ is an $\eps$-relative approximation to $\mu_p(\cT)$.
        \item Output $G \in \cT$ with distribution $\hat \mu$ so that $\| \hat \mu - \mu_{\cT,p} \|_{TV} < \eps$\,.
    \end{enumerate}
\end{theorem}
In the terminology of approximate counting and sampling, these are a Fully Polynomial-time Randomized Approximation Scheme (FPRAS) and an efficient sampling scheme. 

Up to the constant $1/\sqrt{2}$, Theorem~\ref{thmLowMixing} is sharp.  When $p \ge C n^{-1/2}$ for $C$ a large constant, the Glauber dynamics mix slowly due to the emergence of global structure (large max-cuts) in triangle-free graphs. 
\begin{theorem}\label{thmSlowMix}
There exists $C>0$ so that for $p \ge C n^{-1/2}$, the Glauber dynamics for sampling from $\mu_{\cT,p}$ has mixing time $\Omega(e^{\sqrt{n}})$.
\end{theorem}

On the other hand, even in this higher density regime we can still sample from $\mu_{\cT,p}$ and approximate $\mu_p(\cT)$ efficiently.
\begin{theorem}
    \label{thmHighSampling}
    There exist $C,C'>0$, so that for $p \ge C n^{-1/2}$ and $\eps\geq C'e^{-\sqrt{n}}$ there are randomized  algorithms running in time polynomial in $n$ and $1/\eps$ that
    \begin{enumerate}
        \item Output $M$  so that with probability at least $2/3$, $M$ is an $\eps$-relative approximation to $\mu_p(\cT)$.
        \item Output $G \in \cT$ with distribution $\hat \mu$ so that $\| \hat \mu - \mu_{\cT,p} \|_{TV} \le \eps$\,.
    \end{enumerate}
\end{theorem}

Note that Theorem~\ref{thmLowMixing} applies to  all $p\le cn^{-1/2}$, and Theorems~\ref{thmSlowMix} and~\ref{thmHighSampling} apply to all $p \ge C n^{-1/2}$, not just $p = \Theta(n^{-1/2})$.

\subsection{Techniques}
\label{secTechniques}

The techniques we use to prove the main two results (algorithms at low densities,  algorithms at high densities) are related and combine ideas from computer science, statistical physics, and combinatorics.  

Proving the first result is somewhat  straightforward: using the path coupling technique of Bubley and Dyer~\cite{bubley1997path} along with an initial `burn-in' period~\cite{dyer2003randomly} to control the max degree of a graph, we show that the Glauber dynamics for $\mu_{\cT,p}$ are rapidly mixing.

For the second result, efficient sampling at high densities, a different algorithm is needed. In fact, using structural results from~\cite{jenssen2023evolution} we prove that the Glauber dynamics for $\mu_{\cT,p}$ mix slowly when $p\ge C/\sqrt{n}$ (Theorem~\ref{thmSlowMix}).  These structural results (following those in~\cite{luczak2000triangle}) say that almost all $G \sim \mu_{\cT,p}$ have a unique max-cut $(A,B)$ that contains almost all of its edges, and moreover, the graphs induced by $A$ and by $B$ have controlled maximum degree.  In the language of statistical physics we can say almost all $G \sim \mu_{\cT,p}$ are close to some `ground state' collection of bipartite graphs with bipartition $(A,B)$. We call the edges within the parts of a max-cut $(A,B)$ `defect edges'.  While the existence of a large cut is a bottleneck for the Glauber dynamics, we use the same structural result as an algorithmic tool in proving Theorem~\ref{thmHighSampling}, reducing the sampling problem to that of sampling from $\mu_{\cT,p}$ conditioned on a max cut $(A,B)$ and a small max degree among the defect edges. 

With this approach, the main technical step is to show that one can efficiently sample defect edges from (approximately) the correct distribution. There is a rough duality to the problem in that defect edges approximately behave like a \emph{disordered} triangle-free graph $H\sim \mu_{\cT,p}$ where $p=c/\sqrt{n}$ and $c$ is small. However in reality the distribution is significantly more complex. The probability of seeing a set of defect edges $H$ is proportional to a certain hard-core model partition function determined by $H$. Using the machinery of the \textit{cluster expansion} (see Section~\ref{subseccluster}), the defect distribution can  be viewed as an exponential random graph with an unbounded number of subgraph counts in the exponent, conditioned on a max-degree bound and on triangle-freeness. In a similar spirit to Theorem~\ref{thmLowMixing}, we sample defect edges via edge-update Glauber dynamics. However, the implementation is more complicated here. A key step is to use the cluster expansion to accurately estimate edge marginals needed in the implementation of the dynamics. This use of cluster expansion can be viewed as a way to study an exponential random graph as a perturbation of an \ER random graph.

There is a line of previous algorithmic results (e.g.~\cite{HelmuthAlgorithmic2,JKP2,chen2021fast,galanis2021fast,galanis2022fast,helmuth2023finite,carlson2022algorithms,chen2022sampling}) on efficient approximate counting and sampling for spin models (hard-core, Potts, coloring, etc.) on structured instances (expanders, lattices, random graphs) in  low-temperature (i.e.\ strong interaction strength) or high-density regimes that are computationally hard (more precisely \#BIS-hard~\cite{dyer2004relative,cai2016hardness}) in the worst-case.  These algorithms are based on the framework of abstract polymer models in which defects from a ground state configuration or set of configurations are modeled as an auxiliary spin model and then analyzed or sampled from using cluster expansion or Markov chain algorithms.  Other low-temperature algorithms include suitably initialized Markov chains~\cite{gheissari2022low,blanca2023sampling,galanis2023sampling}.  In both cases, it is crucial to the algorithms or to the analysis that defects from the appropriate ground state in the spin model are small: at most logarithmic size.  See the discussion in~\cite[Section 2.2]{chen2022sampling} and the percolation conditions in~\cite{blanca2023sampling}.

What is new in our approach to proving Theorem~\ref{thmHighSampling} is that we can efficiently sample even though the defect edges form connected graphs with large degrees.  At a high level, the novelty of our approach is that in previous works, the algorithms or analyses were based on perturbations around the empty set of defects; here one can interpret our analysis of Glauber dynamics on defect edges as a perturbation around a measure of independent (Erd\H{o}s-R\'{e}nyi) defect edges, and instead of needing to bound the size of their connected components, we need to bound their distance from independence, measured in terms of the mixing time of Glauber dynamics.

\subsection{Open questions}
\label{secOQuestions}

In this paper we only consider sampling triangle-free graphs, but is also natural to consider sampling graphs with forbidden subgraphs $H$ other than triangles. If $H$ is non-bipartite, the `critical density' is determined by the $2$-density of $H$, $m_2(H) = \max_{F \subseteq H; |E(F)| \ge 2} \frac{|E(F)|-1 }{|V(F)|-2  } $.
\begin{question}
    For  non-bipartite graph $H$, can one sample efficiently from $G(n,p)$ conditioned on the event $\{ X_H=0 \}$  in the critical regime $p = \Theta(n^{-1/m_2(H)})$?
\end{question}

 We anticipate that our approach to high-density sampling will be useful more generally in the setting of low-temperature or high-density spin systems mentioned above.  Concretely, we ask if this approach can be used to find efficient algorithms for the hard-core model on random bipartite graphs.
\begin{question}
    Can the high-level approach of the algorithm of Theorem~\ref{thmHighSampling} be used to sample from the hard-core model on random $\Delta$-regular bipartite graphs for all $\lam > \lam_c(\Delta)$?
\end{question}
Currently efficient algorithms are known for $\lam =\Omega(\log \Delta/\Delta)$ from~\cite{chen2022sampling} where the authors point out the barrier of polymers of polynomial size at smaller values of $\lam$.

\subsection{Organization}

In Section~\ref{secPrelim} we provide some preliminaries on Markov chain mixing and the cluster expansion.  In Section~\ref{secMixLow} we prove Theorem~\ref{thmLowMixing} via a path coupling argument with burn-in.     
In Section~\ref{secSampleHigh} we prove Theorem~\ref{thmHighSampling}, showing the existence of efficient sampling and counting algorithms at high densities. 
In Section~\ref{secSlowMix} we prove Theorem~\ref{secSlowMix}, showing that the Glauber dynamics for sampling from $\mu_{\cT,p}$ mixes slowly for large $p$.

\section{Preliminaries}
\label{secPrelim}

\subsection{Notation}
For a graph $G = (V,E)$ we denote the number of edges by $|G|$ and its maximum degree by $\Delta(G)$. All graphs in this paper will have $n$ vertices unless specified otherwise.  For $v\in V$ we let $d_G(v)$ denote the degree of the vertex $v$ and we let $d_G=2|G|/n$ denote the average degree. 
We use $\mu_p$ to refer to the distribution of the \ER random graph $G(n,p)$; that is, $\mu_p(G) = p^{|G|}(1-p)^{\binom{n}{2} -|G|}$. We let $\cT$ denote the set of triangle-free graphs on $n$ vertices or the event that a random graph is triangle-free.  The conditional distribution $\mu_p(\cdot | \cT)$ is denoted $\mu_{\cT,p}$.   The partition function $Z(\lam) = \sum_{G \in \cT} \lam^{|G|}$ will always be used with $\lam = p/(1-p)$ giving the identity $\mu_p(\cT) = (1-p)^{\binom{n}{2}} Z(\lam)$.   

\subsection{Markov chains and mixing times}
\label{secMCMC}

Our sampling algorithms for both the low density and high density cases will use Markov chains.  In the low density case, the Markov chain approach will directly give an efficient sampling algorithm, while in the high density case, Markov chains will be one part of a more complicated algorithm.

The basic idea to approximately sample from a target distribution $\mu$  is to design a Markov chain with $\mu$ as the stationary distribution so that a single step of the chain can be implemented efficiently, and so that the chain converges quickly.   The convergence time is often quantified by the \textit{mixing time}.  With $\mu_t^{X_0}$ denoting the $t$-step distribution of the Markov chain starting from the state $X_0$, the mixing time is 
\begin{equation}
    \tau_{\mathrm{mix}} = \max_{X_0} \, \min  \left \{ t: \left \| \mu_t^{X_0} - \mu  \right \|_{TV} < 1/4 \right \} \,.
\end{equation}
Here $\| \cdot - \cdot \|_{TV}$ denotes the total variation distance between probability measures on the same space, and the choice of $1/4$ is arbitrary and can be reduced to any $\eps > 0$ by running  the chain a factor $O(\log(1/\eps))$ more than $\tau_{\mathrm{mix}}$.  For background on Markov chains and mixing times see~\cite{montenegro2006mathematical,levin2017markov}.

We will consider Markov chains  on graphs;  these will have the form of the \textit{Glauber dynamics}: at each step we choose a potential edge $e \in \binom{V}{2}$ uniformly at random and then resample it (whether it is in the graph or not) conditioned on the status of all other edges.  More generally, the Glauber dynamics can be applied to any high-dimensional probability distribution by choosing a random coordinate and resampling the value of that coordinate conditioned on the rest of the vector.

For the \ER distribution $G(n,p)$ the Glauber dynamics are particularly simple: at each step one of the $\binom{n}{2}$ possible edges is chosen uniformly; with probability $p$ the edge is included in $G$ and with probability $1-p$ it is not included.   

For $\mu_{\cT,p}$, the Glauber dynamics are similar: at each step one of the $\binom{n}{2}$ possible edges is chosen uniformly; with probability $p$ the edge is included in $G$, but only if its inclusion would not form a triangle with other edges already present.

\subsection{Cluster expansion and the hard-core model}\label{subseccluster}
 
Recall $Z_G(\lam)$, the partition function of the hard-core model on a graph $G$ from~\eqref{eqHCzdef}. 
The cluster expansion is a formal power series for $\log Z_G(\lam)$; in fact, it is the Taylor series around $\lam =0$.  Conveniently, the terms of the cluster expansion have a nice combinatorial interpretation (see e.g.~\cite{scott2005repulsive,faris2010combinatorics}). A \textit{cluster} $\Gamma=(v_1, \ldots, v_k)$ is a tuple of vertices from $G$ such that the induced graph $G[\{v_1, \ldots, v_k\}]$ is connected. We let $\cC(G)$ denote the set of all clusters of $G$. We call $k$ the size of the cluster and denote it by $|\Gamma|$. Given a cluster $\Gamma$, the \emph{incompatibility graph} $G_\Gamma$ is the graph on vertex set $\Gamma$ (considered as a multiset) with an edge between $v_i, v_j$ if either $v_i, v_j$ are adjacent on $G$ or $i \ne j$ and $v_i, v_j$ correspond to the same vertex in $G$.   In particular, by the definition of a cluster, the incompatibility graph $G_\Gamma$ is connected.

As a formal power series, the cluster expansion is the infinite series
\begin{align}\label{eqclusterexp}
\log Z_G(\lam) = \sum_{\Gamma\in \cC(G)} \phi_G(\Gamma) \lam^{|\Gamma|} \,,
\end{align}
where
\begin{align}
\label{eqUrsell}
\phi_G(\Gamma) &= \frac{1}{|\Gamma|!} \sum_{\substack{A \subseteq E(G_\Gamma)\\ \text{spanning, connected}}}  (-1)^{|A|} \, .
\end{align}
If the graph $G$ is clear from the context we will often write $\phi(\Gamma)$ in place of $\phi_G(\Gamma)$.

The cluster expansion converges absolutely if $\lam$ lies inside a disk $D \subset \mathbb C$ so that $Z_G(\xi) \ne 0$ for all $\xi \in D$.  We will use the following lemma, \cite[Lemma 4.1]{jenssen2023evolution}, which gives a sufficient condition for convergence and  bounds the error in truncating the cluster expansion.

\begin{lemma}\label{lemClusterTail}
Suppose $G$ is a graph on $n$ vertices with maximum degree $\Delta$, and suppose $ |\lambda|\leq \frac{1}{2e(\Delta+1)}$. Then the cluster expansion converges absolutely. Moreover, for any set $U\subseteq V(G)$ such that $1\leq |U|\leq \min\{2,k\}$
we have
   \begin{align*}
  \sum_ {\substack{\Gamma: \Gamma\supseteq U, \\ |\Gamma|\geq k}} |\phi(\Gamma)||\lam|^{|\Gamma|}
 \leq 
 (2e)^k\Delta^{k-|U|}|\lam|^k\, .
 \end{align*}
\end{lemma}
We remark that the conclusion of \cite[Lemma 4.1]{jenssen2023evolution} is slightly weaker than the statement above (the absolute value signs appear outside of the sum), however, it is readily checked that the proof from  \cite{jenssen2023evolution} gives the statement of Lemma~\ref{lemClusterTail}.

\section{Sampling at low density}
\label{secMixLow}
In this section we prove Theorem~\ref{thmLowMixing}.  Our proof will proceed by a path coupling argument where we use an initial burn-in period to make sure that the graphs in the coupling have reasonable maximum degree. To this end it will be useful to observe  that $\mu_{\cT,p}$ is stochastically dominated by $\mu_p$.

\begin{lemma}\label{lemStochDom}
For every $p\in[0,1]$, the measure $\mu_p$ stochastically dominates $\mu_{\cT,p}$.  That is, there is a coupling of the two distributions so that with probability $1$, $G \subseteq G'$, where $G \sim \mu_{\cT,p}$ and $G' \sim \mu_p$.
\end{lemma}

\begin{proof}
Sample $G \sim\mu_{\cT,p}$ by sampling one edge at a time given the previous history.  At any step, the conditional probability of any edge is at most $p$. We can therefore couple the sampling with an edge-by-edge sampling of $G(n,p)$ so that an edge is present in the sample from $\mu_{\cT,p}$ only if it is present in the sample from $\mu_p$.
\end{proof}

This fact combined with Chernoff's inequality and a union bound over vertices yields the following corollary which will allow us to condition on the event ${\Delta(G)\leq (1+\eps)np}$.

\begin{cor}\label{cormaxdegdom}
Let $p,\eps\in [0,1]$ and let $G\sim \mu_{\cT, p}$. Then with $\mu=(n-1)p$ 
\[
\P(\Delta(G)\geq (1+\eps)\mu) \leq n e^{-\eps^2 \mu/3}\, .
\]
\end{cor}

It will also be useful to observe that we can couple the triangle-free Glauber dynamics for $\mu_{\cT,p}$ with the Glauber dynamics for $\mu_p$, in which one of the $\binom{n}{2}$ potential edges is chosen at random and resampled at each step.

\begin{lemma}\label{lemcoupleGnp}
There exists a coupling $(X_t, Y_t)_{t\geq 0}$ of the Glauber dynamics for $\mu_{\cT,p}$ with the Glauber dynamics for $\mu_p$ such that if $X_0$ is a subgraph of $Y_0$, then $X_t$ is a subgraph of $Y_t$ for all $t\geq 0$.
\end{lemma}
\begin{proof}
We make the same edge updates in both chains, unless a pair $\{u,v\}\in \binom{[n]}{2}$ is chosen such that $u, v$ have a common neighbour in $X_t$, in which case we set $X_{t+1}=X_t$.
\end{proof}

A coupon collector argument gives us that the mixing time for Glauber dynamics for $\mu_p$ is $O(n^2\log n)$. As a result, we have the following corollary. 

\begin{cor}\label{cordegreemix}
There exists $C>0$ so that the following holds for any $c>0$.  Let $p \leq c/\sqrt{n}$.
Let $(X_t)_{t\geq 0}$ be a run of the Glauber dynamics for $\mu_{\cT,p}$ on $n$ vertices with $X_0$ arbitrary.  If $t\geq C n^2 \log(n/\eps)$,
\[
\P(\Delta(X_t)\geq np + n^{1/3} )\leq ne^{-n^{2/3}/(3np)}+\eps\, .
\]
\end{cor}
\begin{proof}
Consider the coupling $(X_t, Y_t)_{t\geq 0}$ from Lemma~\ref{lemcoupleGnp} with $X_0=Y_0$.
Let $Y$ be sampled from $\mu_p$. For $t\geq C n^2 \log(n/\eps)$ we have
\[
\|Y-Y_t\|_{TV}\leq \eps\, .
\]
By the Chernoff bound and the union bound, we have $\mathbb{P}(\Delta(Y) \geq np+n^{1/3})\leq ne^{-n^{2/3}/(3np)}$.
So it follows that $\mathbb{P}(\Delta(Y_t) \geq np+n^{1/3}) \leq ne^{-n^{2/3}/(3np)}+\eps$. Since $X_t$ is a subgraph of $Y_t$ the result follows. 
\end{proof}

Now let $\Omega(p)\subseteq \cT$ be the set of graphs in $\cT$ whose maximum degree is at most $np + n^{1/3}$.

\begin{lemma}\label{lempathcouple}
Let $p=(1+o(1)) c/\sqrt{n}$ where $c<1/\sqrt{2}$ is fixed. Let $(X_t, Y_t)\in \Omega(p)\times \Omega(p)$ and let $(X_{t+1}, Y_{t+1})$ be copies of the Glauber dynamics for $\mu_{\cT,p}$ coupled to attempt the same updates. Then
\[
\E \left[d (X_{t+1}, Y_{t+1})\right]\leq \left(1- \frac{\delta}{\binom{n}{2}}\right) d(X_t, Y_t)\, ,
\]
where $\delta = 1 - 2(np+n^{1/3})p = 1-2c^2 - o(1)$ and $d(\cdot,\cdot)$ is the Hamming distance between graphs.
\end{lemma}
\begin{proof}
By the path coupling technique~\cite{bubley1997path}, it suffices to consider the case where $d(X_t, Y_t)=1$ i.e. $X_t, Y_t$ differ in one edge. 
Suppose WLOG that $\{i,j\}$ is an edge of $X_t$, but not $Y_t$.
Let $d_i, d_j$ denote the degree of $i,j$ in $Y_t$ respectively and note that $d_i, d_j \leq np+n^{1/3}$. If the update was performed on $\{i,j\}$, then $X_{t+1} = Y_{t+1}$. If the update was performed on one of the $d_i+d_j-2$ potential edges that form a triangle together with $\{i,j\}$, then this edge would be added to $Y_{t+1}$ but not to $X_{t+1}$ with probability at most $p$. If the update is performed on any other potential edge, then $d(X_{t+1},Y_{t+1})=d(X_t,Y_t)$. Thus
\[
\E \left[d (X_{t+1}, Y_{t+1})\right] \leq 1+\frac{d_i+d_j}{\binom{n}{2}}\cdot p - \frac{1}{\binom{n}{2}}\leq 1- \frac{\delta}{\binom{n}{2}}\, ,
\]
as claimed.
\end{proof}

We are now ready to prove Theorem~\ref{thmLowMixing}.

\begin{proof}[Proof of Theorem~\ref{thmLowMixing}]
Let $(X_t, Y_t)_{t\geq 0}$ be a coupling of the Glauber dynamics for $\mu_{\cT,p}$ to attempt the same updates with $X_0, Y_0$ chosen arbitrarily.
Now let $T=6C n^2 \log n$ where $C$ is as in Corollary~\ref{cordegreemix}.

For $t\geq T$, let $\mathcal E_t$ denote the event $\{X_{t},Y_{t} \in \Omega(p)\}$ and note that by Corollary~\ref{cordegreemix} (with $\eps=n^{-5}$) and a union bound we have $\P(\mathcal E_t^c)\leq 2(ne^{-n^{2/3}/(3np)}+n^{-5})=O(n^{-5})$. We then have,
\begin{align*}
\P(X_{2T}\neq Y_{2T})&\leq \E\left[d (X_{2T}, Y_{2T}) \right]\\
& \leq \E\left[d (X_{2T}, Y_{2T}) \mid \mathcal E_{2T-1} \right] + O(n^{-3})\\
&\leq \left(1- \frac{\delta}{\binom{n}{2}}\right) \E\left[d (X_{2T-1}, Y_{2T-1})\right] + O(n^{-3}).
\end{align*}
For the first inequality we used Markov. For the second inequality we used that $d (X_{2T}, Y_{2T})\leq n^2$ and $\P(\mathcal E_{2T-1}^c)=O(n^{-5})$.
For the third inequality we used Lemma~\ref{lempathcouple}. 
Iterating the above $T-1$ more times, we conclude that
\begin{align*}
\P(X_{2T}\neq Y_{2T}) &\leq \left(1- \frac{\delta}{\binom{n}{2}}\right)^T \E\left[d (X_{T}, Y_{T})\right] + O(Tn^{-3})\\
&=o(1)\, ,
\end{align*}
where for the final inequality we used that $\E\left[d (X_{T}, Y_{T})\right]\leq n^2$.

It follows that 
\[
\tau_{\text{mix}}\leq 2T \, .
\]
The existence of an efficient sampling scheme follows immediately, since a single step of the Glauber dynamics can be implemented efficiently.

The existence of an FPRAS for approximating $\mu_p(\cT)$ follows from a standard reduction of approximate counting to sampling. Using the adaptive simulated annealing technique from~\cite{vstefankovivc2009adaptive}, which requires an efficient approximate sampling algorithm for $\mu_{\cT,p'}$ for all $0 \le p'\le p$, one obtains an FPRAS for $\mu_p(\cT)$ running in time $O(n^3 \eps^{-2} \log ^6 n)$.  The algorithm proceeds by using samples from $\mu_{\cT,p'}$ to estimate the ratio $\mu_{p''}(\cT)/\mu_{p'}(\cT)$ for adaptively chosen pairs $p',p''$ along a `cooling schedule', then multiplying these estimates in a telescoping product. 
\end{proof}

\section{Overview for sampling at high densities}
\label{secSampleHigh}

In this section we give an overview of the proof of Theorem~\ref{thmHighSampling}, with the remaining proof details to follow in Section~\ref{secProofsHighDensity}. The proof relies on a structural result for graphs drawn from $\mu_{\cT,p}$ established in~\cite{jenssen2023evolution} in the regime where $p=C/\sqrt{n}$ and $C$ is large.
To state the result we introduce some notation and definitions from~\cite{jenssen2023evolution}. 

 \begin{defn}\label{defWeakBalance}
Call a partition $(A,B)$ of $[n]$ \emph{weakly balanced} if $\big ||A|-|B| \big |\leq n/10$. Let $\Pi_{\textup{weak}}$ denote the set of weakly balanced partitions of $[n]$.
\end{defn}

We set $\alpha=\frac{1}{96e^3}$. This constant is taken from~\cite{jenssen2023evolution}. Its precise form does not matter, only that it is sufficiently small. For a partition $(A,B)$ of $[n]$ let
\begin{align}\label{eqTABwDef}
\cT_{A,B,\lam}^{\textup{w}}=\{G\in \cT : \Delta(G[A]\cup G[B])\leq \alpha/ \lam\}\, ,
\end{align}
where $G[A], G[B]$ denote the induced subgraphs of $G$ on vertex sets $A,B$ respectively.

We think of the set $\cT_{A,B,\lam}^{\textup{w}}$ as those graphs $G\in \cT$ which `align well' with the bipartition $(A,B)$ (i.e.\ $(A,B)$ is a large cut for $G$).

 Let
 \begin{equation}
    Z_{A,B}^\textup{w}(\lam) : = \sum_{G\in\cT_{A,B,\lam}^{\textup{w}}}\lam^{|G|}\, .
 \end{equation}
In other words, $Z_{A,B}^\textup{w}(\lam)$ is the contribution to the partition function $Z(\lam)=\sum_{G\in\cT}\lam^{|G|}$ from graphs that align with $(A,B)$.

Now define a distribution $\mu_{\mathrm{weak},\lam}$ on $\cT$ as follows.
\begin{enumerate}
    \item Sample $(A,B)\in \Pi_{\textup{weak}}$ with probability proportional to $Z_{A,B}^{\textup{w}}(\lam)$.
    \item Sample $G\in  \cT_{A,B,\lam}^\textup{w}$ with probability proportional to $\lam^{|G|}$.\label{S12}
\item Output $G$.
\end{enumerate}

Let
\[
Z_{\textup{weak}}(\lam) = \sum_{(A,B)\in\Pi_{\textup{weak}}} Z_{A,B}^{\textup{w}}(\lam)\, .
\]
Throughout this section we set $C>0$ to be a sufficiently large constant.
\begin{theorem}[{\cite[Theorem 2.9]{jenssen2023evolution}}]
    For $\lam \ge C n^{-1/2}$, 
    \[ \| \mu_{\cT,p} - \mu_{\mathrm{weak},\lam}  \|_{TV} = o(1) \,,\]
    where $p=\frac{\lam}{1+\lam}$.
\end{theorem}
Here we make the error estimate explicit.
\begin{prop}
    \label{PropError}
    For $\lam \ge C n^{-1/2}$,
     \[ \|  \mu_{\cT,p} - \mu_{\mathrm{weak},\lam}  \|_{TV} = O(e^{-\sqrt{n}}) \,,\]
     where $p=\frac{\lam}{1+\lam}$. Moreover
     \[
     Z(\lam)=(1+O(e^{-\sqrt{n}}))Z_{\textup{weak}}(\lam)\, .
     \] 
\end{prop}
We defer the proof of Proposition~\ref{PropError} to Section~\ref{secMuWeakError}. 

With Proposition~\ref{PropError} in hand, we turn our attention to sampling from the measure $\mu_{\mathrm{weak},\lam}$. For this, we split Step~\ref{S12} in the definition of $\mu_{\mathrm{weak},\lam}$ into two substeps: first we sample the `defect edges' within parts $A$ and $B$ and then we sample edges crossing the partition $(A,B)$. The second of these substeps will be relatively simple. As we will see, the crossing edges have the distribution of a (graph) hard-core model, and the condition in~\eqref{eqTABwDef} ensures that this model is subcritical in the sense that $\lam$ is  small enough as a function of the max degree that the cluster expansion converges and efficient sampling algorithms exist. Most of the work  in proving Theorem~\ref{thmHighSampling} will come from implementing the first step and showing we can sample defect edges efficiently. To make this more precise, we make some further definitions. 

 Given a partition $(A,B)$ of $[n]$ and a pair of graphs $S\subseteq \binom{A}{2}, T\subseteq \binom{B}{2}$, we write $S\boxempty T$ as a shorthand for the Cartesian product of the graphs $(A,S), (B,T)$, i.e., the graph with vertex set $V(S \boxempty T) = A \times B$ and edge set $E(S \boxempty T)$ equal to \[ \{ \{(a,b), (a,b')\}: \{b,b'\}\in T  \} \cup \{ \{(a,b), (a',b)\}: \{a,a'\}\in S  \} \,.\]

The significance of the Cartesian product for us comes from the following lemma.
\begin{lemma}
\label{lemSTgraphLemma}
Let $(A, B)$ be a partition of $[n]$ and suppose $S\subseteq \binom{A}{2}, T \subseteq \binom{B}{2}$ such that $S\cup T$ is triangle-free.  Let $\cG(S,T)$ be the set of triangle-free graphs $G$ so that $G[A]= S$ and $G[B]= T$.  Then
\[ \sum_{G \in \cG(S,T)} \lam^{|G|} =  \lam^{|S| + |T|} Z_{S\boxempty T} (\lam) \,,\]
where $Z_{S\boxempty T} (\lam)$ is the hard-core partition function on the graph $S\boxempty T$.
\end{lemma}
\begin{proof}
In what follows we identify the vertex $(u,v)\in V(S\boxempty T)=A \times B$ with the edge $\{u,v\}\in \binom{[n]}{2}$.
The proof follows from the observation that if $I$ is an independent set in the graph $S\boxempty T$, then  the graph on $[n]$ with edge set $S \cup T \cup I$ is a triangle-free graph in $\cG(S,T)$, and likewise for any $G \in \cG(S,T)$, $E(G) \cap (A\times B)$ forms an independent set in $S\boxempty T$, giving a  one-to-one correspondence.
\end{proof}

Let $\cD_{A,B,\lam}^\textup{w}$ denote the set of pairs $(S,T)$ such that $S\subseteq\binom{A}{2}$, $T\subseteq\binom{B}{2}$, $\Delta(S\cup T)\leq \alpha/\lam$ and $S\cup T$ is triangle free. In other words, 
\begin{align}\label{eqDABwDef}
 \cD_{A,B,\lam}^\textup{w}= 
 \{(G[A], G[B]): G\in  \cT_{A,B,\lam}^{\textup{w}}\}\, ,
\end{align}
the set of possible defect graphs with respect to $(A,B)$. By Lemma~\ref{lemSTgraphLemma} we see that
\[
Z_{A,B}^{\textup{w}}(\lam) = \sum_{(S,T)\in \cD_{A,B,\lam}^\textup{w}} \lam^{|S|+|T|} Z_{S\boxempty T} (\lam)\, ,
\]
and we may rewrite the definition of $\mu_{\mathrm{weak},\lam}$ as
 \begin{enumerate}
     \item Sample $(A,B)\in\Pi_{\text{weak}}$ with probability proportional to $Z_{A,B}^{\textup{w}}(\lam)$.\label{S1}
     \item Sample $(S,T)\in  \cD_{A,B,\lam}^\textup{w}$ with probability proportional to $\lam^{|S|+|T|}Z_{S\boxempty T}(\lam)$.\label{S2}
     \item Sample $\Ec\subseteq A\times B$ from the hard-core model on $S\boxempty T$ at activity $\lam$. \label{S3}
 \item Let $E= S\cup T\cup \Ec$ and output the graph $G= ([n], E)$.
\end{enumerate}
As mentioned above, Step~\ref{S2} will be our main focus. Given a weakly balanced partition $(A,B)$ let $\nu_{A,B,\lam}$ be the distribution on $\cD_{A,B,\lam}^\textup{w}$ defined by
\begin{equation}
    \nu_{A,B,\lam} (S,T) = \frac{ \lam^{|S|+|T|}Z_{S\boxempty T}(\lam)  } { Z_{A,B}^{\textup{w}}(\lam)  } \,.
\end{equation}

 Our main task is to show that we can sample from $\nu_{A,B,\lam} $ efficiently. First we show that for $\lam \ge C n^{-1/2}$, the edge-update Glauber dynamics for $\nu_{A,B,\lam}$ mixes rapidly. 
 \begin{prop}
    \label{propNuMix}
    For $\lam \ge C n^{-1/2}$, the mixing time of the edge-update Glauber dynamics for $\nu_{A,B,\lam}$ is $O(n^2 \log n)$,
    uniformly over weakly balanced $(A,B)$.
\end{prop}

We prove Proposition~\ref{propNuMix} in the next section. We use path coupling in a similar spirit to the proof of Theorem~\ref{thmLowMixing}, however now the task is significantly more difficult due to the complex interactions between edges under the measure $\nu_{A,B,\lam}$. In particular, if we condition on the entire graph outside of the edge $e$, computing the conditional marginal of $e$ to sufficient accuracy is a non-trivial task (whereas this is trivial for $\mu_{\cT,p}$). Our proof of Proposition~\ref{propNuMix} leans heavily on the cluster expansion to estimate these edge marginals.  

Once we have proved Proposition~\ref{propNuMix}, we show that we can approximately implement the steps of the Glauber dynamics efficiently to obtain a sampling algorithm for $\nu_{A,B,\lam}$. 
\begin{prop}\label{propDefectSample}
If $\lam>C/\sqrt{n}$ and $(A,B)$ is weakly balanced then there are randomized algorithms running in time polynomial in $n$, $1/\eps$ and $\log(1/\delta)$ (uniformly over all $(A,B)$) that
  \begin{enumerate}
        \item Output $M$  so that with probability at least $1-\delta$, $M$ is an $\eps$-relative approximation to $Z_{A,B,\lam}^{\textup{w}}(\lam)$.
        \item Output $G \in \cT$ with distribution $\hat\nu$ so that $\| \hat \nu - \nu_{A,B,\lam} \|_{TV} \le \eps$\,.
    \end{enumerate}
\end{prop}
We prove Proposition~\ref{propDefectSample} in the next section.  

\section{Proofs for sampling at high density}
\label{secProofsHighDensity}
In this section we prove Propositions~\ref{propNuMix} and~\ref{propDefectSample} and deduce Theorem~\ref{thmHighSampling}. We begin with Proposition~\ref{propNuMix}.

\subsection{Glauber dynamics for the defect distribution}

A key step in the proof of Proposition~\ref{propNuMix} is to show that $\nu_{A,B,\lam}$ is stochastically dominated by two independent copies of relatively sparse \ER random graphs on $A,B$ respectively. Given a set $V$ and $q\in[0,1]$, we let $\mu(V,q)$ denote the measure associated to the \ER random graph on vertex set $V$ and edge probability $q$.
\begin{lemma}\label{lemSDlargeC}
For $\lam\geq Cn^{-1/2}$, the measure $\nu_{A,B,\lam}$ is stochastically dominated by $\mu(A,q)\times \mu(B,q)$ where $q= \lam e^{-n\lam^2/5}$. 
\end{lemma}
Lemma~\ref{lemSDlargeC} shows that as $\lam$ gets larger samples from the defect measure $\nu_{A,B,\lam}$ become increasingly sparse. 
We will prove Lemma~\ref{lemSDlargeC} in the following subsection (Section~\ref{subsecSD}). We record the following corollary.

\begin{cor}\label{cordegreemixsuper}
Let $\lam\geq Cn^{-1/2}$.
There exists $C'>0$ so that the following holds.
Let $(X_t)_{t\geq 0}$ be a run of the Glauber dynamics for $\nu_{A,B,\lam}$ with $X_0$ arbitrary.  If $t\geq C' n^2 \log(n/\eps)$,
\[
\P(\Delta(X_t)\geq 2nq )\leq ne^{-nq/3}+\eps\, 
\]
where $q= \lam e^{-n\lam^2/5}$.
\end{cor}

We will prove Proposition~\ref{propNuMix} via path coupling. For $(\mathbf{S},\mathbf{T})\sim \nu_{A,B,\lam}$,  $S\subseteq \binom{A}{2}$, $T\subseteq \binom{B}{2}$ and $e\in \binom{A}{2}\cup \binom{B}{2}$, define the conditional edge marginal
\[
p(e | S, T):=\P(e\in \mathbf{S}\cup \mathbf{T} \mid (\mathbf{S}\cup \mathbf{T})\backslash e=S\cup T)\, .
\]
If $\P((\mathbf{S}\cup \mathbf{T})\backslash e=S\cup T)=0$ then we define $p(e | S, T)$ to be $0$.

Let $(X_t,Y_t)_t\geq0$, denote two runs of the Glauber dynamics of $\nu_{A,B,\lam}$ coupled so that given $t\geq 0$ and $(X_t,Y_t)$ we perform the following update: 
\begin{itemize}
\item Pick $e\in\binom{A}{2}\cup \binom{B}{2}$ uniformly at random and let $U\sim U[0,1]$.
\item If $U\leq p(e | X_t)$ then set $X_{t+1}=X_t\cup e$. Otherwise set $X_{t+1}=X_t\backslash e$. 
\item Similarly, if $U\leq p(e | Y_t)$ then set $Y_{t+1}=Y_t\cup e$. Otherwise set $Y_{t+1}=Y_t\backslash e$. 
\end{itemize}

We refer to this as the \emph{optimal} coupling. Let
\begin{align}\label{eqOmegaDef}
\Omega(\lam)=\Bigg\{(S,T)\in\cD_{A,B,\lam}^{\textup{w}}: \Delta(S\cup T)\leq 2n\lam e^{-n\lam^2/5} \Bigg\}\, .
\end{align}

In analogy to  Lemma~\ref{lempathcouple}, our goal is to prove the following path coupling result. 
\begin{lemma}\label{lempathcoupleNu}
Let $\lam\geq Cn^{-1/2}$. Let $(X_t, Y_t)\in \Omega(\lam)\times \Omega(\lam)$ and let $(X_{t+1}, Y_{t+1})$ be copies of the optimally coupled Glauber dynamics for $\nu_{A,B,\lam}$. Then
\[
\E \left[d (X_{t+1}, Y_{t+1})\right]\leq \left(1- \frac{1}{n^2}\right) d(X_t, Y_t)\, ,
\]
where $d(\cdot,\cdot)$ is the Hamming distance between graphs.
\end{lemma}

 We note that Proposition~\ref{propNuMix} follows from Lemma~\ref{lempathcoupleNu} via the same burn-in argument that was used to deduce Theorem~\ref{thmLowMixing} from Lemma~\ref{lempathcouple} (only now we use Corollary~\ref{cordegreemixsuper} in place of Corollary~\ref{cordegreemix}).

To prove Lemma~\ref{lempathcoupleNu} we need a good understanding of the edge marginals that appear in the definition of the optimal coupling above. This will be a goal of the next subsection.

\subsection{Stochastic domination and computing edge marginals in $\nu_{A,B,\lam}$.}\label{subsecSD}

Throughout this section we fix a weakly balanced partition $(A,B)$. For the proofs to come, it will be convenient to use the following shorthand. Given a graph $G\subseteq \binom{A}{2}\cup \binom{B}{2}$ we let $G_\boxempty$ denote the graph $G[A]\boxempty G[B]$. We will identify $G$ with the pair $(G[A], G[B])$ and write $\nu_{A,B,\lam}(G)$ in place of $\nu_{A,B,\lam}(G[A],G[B])$ and $p(e|G)$ in place of $p(e|G[A], G[B])$.

Lemma~\ref{lemSDlargeC} follows immediately from the following. 
\begin{lemma}\label{lemSDedgemarg}
Let $\lam\geq Cn^{-1/2}$. For $e\in \binom{A}{2}\cup \binom{B}{2}$, $G\subseteq \binom{A}{2}\cup \binom{B}{2}$,
\[
p(e | G) \leq \lam e^{-n\lam^2/5}\, .
\]
\end{lemma}
\begin{proof}
We may assume that $\nu_{A,B,\lam}(G \cup e)>0$ and $e \not\in G$ else the result is trivial. In this case,
\begin{align}\label{eqSDbyLocal}
p(e | G)
&=\frac{\nu_{A,B,\lam}(G \cup e)}{\nu_{A,B,\lam}(G) + \nu_{A,B,\lam}(G \cup e)}
\\
&=
 \frac{\lam Z_{(G\cup e)_\boxempty}(\lam)/Z_{G_\boxempty }(\lam)}{1+\lam Z_{(G\cup e)_\boxempty}(\lam)/Z_{G_\boxempty}(\lam)}\, .
\end{align}
Let us estimate the numerator in the expression above. 
First note that $\Delta((G\cup e)_\boxempty )\leq 2\Delta(G\cup e)\leq 2\alpha/\lam$ since $G\cup e\in \cD_{A,B,\lam}^{\textup{w}}$ by assumption. Since $\lam\leq (2e(1+\alpha/\lam))^{-1}$, we may apply Lemma~\ref{lemClusterTail} (the cluster expansion). 

Let $\cC'(G_\boxempty)$ denote the set of non-constant clusters of $G_\boxempty$ and let
\begin{align}\label{eqcCdprime}
\cC''=\cC'((G\cup e)_\boxempty) \backslash \cC'(G_\boxempty)\, .
\end{align}
Assume WLOG that $e\in\binom{A}{2}$, then 
\begin{align}
\log \left( \frac{Z_{(G\cup e)_\boxempty }(\lam)}{Z_{G_\boxempty}(\lam)} \right)&= \sum_{\Gamma\in \cC''}\phi(\Gamma)\lam^{|\Gamma|}\nonumber\\
&= -\lam^2|B|+ \sum_{\Gamma\in \cC'': |\Gamma|\geq 3}\phi(\Gamma)\lam^{|\Gamma|}\, .\label{eqLkDiff}
\end{align}
For the second equality we used that there are no clusters of size one in $\cC''$ and the clusters of size two are of the form $(\gamma_1, \gamma_2)$ where $\{\gamma_1, \gamma_2\}$ is an edge of $(G\cup e)_\boxempty $ but not of $G_\boxempty$, and there are $|B|$ such edges. Call such an edge $\{\gamma_1, \gamma_2\}$ an \emph{$e$-edge}. 
Now note that if $\Gamma\in \cC''$, then $\Gamma$ must contain an $e$-edge $\{\gamma_1, \gamma_2\}$. By Lemma~\ref{lemClusterTail} (applied with $k=3$ and $U=\{\gamma_1, \gamma_2\}$ an $e$-edge) we have
\begin{align*}
\left|  \sum_ {\Gamma\in \cC'': |\Gamma|\geq 3} \phi(\Gamma)\lam^{|\Gamma|} \right| 
&\leq |B|(2e)^3\Delta((G\cup e)_\boxempty )\lam^3\\
&\leq |B|2^4e^3 \alpha\lam^2\\
&\leq \lam^2|B|/2
\, ,
\end{align*}
where for the second inequality we used that $\Delta((G\cup e)_\boxempty )\leq 2\Delta(G\cup e)\leq 2\alpha/\lam$ since $G\cup e\in \cD_{A,B,\lam}^{\textup{w}}$ by assumption.
 Returning to~\eqref{eqLkDiff} we conclude that
\begin{align}\label{eqLkDiffconc}
   \log \left( \frac{Z_{(G\cup e)_\boxempty }(\lam)}{Z_{G_\boxempty}(\lam)} \right)\leq - \lam^2|B|/2
\end{align}
and so by~\eqref{eqSDbyLocal}
\[
p(e | G) 
\leq
\lam e^{-|B|\lam^2/2}\leq \lam e^{-n\lam^2/5} \, ,
\]
where for the last inequality we used that $(A,B)$ is weakly balanced. 
\end{proof}

With Lemma~\ref{lempathcoupleNu} in mind, we now prove a result that compares the marginals of an edge $e$ in two configurations that differ in exactly one edge $f$. First we introduce some notation.

For $e\in \binom{A}{2}$, we call a pair in $A\times B$ of the form $\{(u,x), (v,x)\}$ an \emph{$e$-pair} (and similarly if $e\in \binom{B}{2}$). For $e,f\in \binom{A}{2}\cup \binom{B}{2} $, let $\cC_{e,f}(G_\boxempty)$ denote the set 
\[
\{\Gamma\in \cC(G_\boxempty): \Gamma \text{ contains both an $e$-pair and an $f$-pair}\}\, .
\]
Recall the definition of $\Omega(\lam)$ from~\eqref{eqOmegaDef}.
\begin{lemma}\label{lemMargDiff}
Let $\lam\geq Cn^{-1/2}$. Let $e, f\in \binom{A}{2} \cup \binom{B}{2}$, $G \in \Omega(\lam)$.  If  $p(e | G \cup f)>0$, then
\begin{align*}
 |p(e | G) - p(e | G \cup f)|\leq 2\lam e^{-\lam^2n/5} R
 \end{align*}
 where 
 \begin{align*}
R:= 
\sum_{\substack{\Gamma\in \cC_{e,f}((G\cup 
 e)_\boxempty)}}|\phi(\Gamma)|\lam^{|\Gamma|}
 +
 \sum_{\substack{\Gamma\in \cC_{e,f}((G\cup 
 e\cup f)_\boxempty)}}|\phi(\Gamma)|\lam^{|\Gamma|}
\, .
\end{align*}
\end{lemma}
\begin{proof}
We continue with the notation in the proof of Lemma~\ref{lemSDedgemarg}, with the additional assumption that $f \not\in G$.
Let $\cC''$ be as in~\eqref{eqcCdprime} and let
\[
\tilde \cC:= \cC'((G\cup e \cup f)_\boxempty) \backslash \cC'((G \cup f)_\boxempty)
\]
By~\eqref{eqLkDiff} we have
\begin{align}
&\log \left( \frac{Z_{(G\cup e)_\boxempty }(\lam)}{Z_{G_\boxempty}(\lam)} \right)- \log \left( \frac{Z_{(G\cup e\cup f)_\boxempty}(\lam)}{Z_{(G \cup f)_\boxempty}(\lam)} \right)\nonumber\\
&=  \sum_{\Gamma\in \cC''}\phi_{1}(\Gamma)\lam^{|\Gamma|} 
-
 \sum_{\Gamma\in \tilde \cC}\phi_{2}(\Gamma)\lam^{|\Gamma|} \, ,\label{eqPhi12}
\end{align}
where $\phi_1=\phi_{(G\cup e)_\boxempty}$ and $\phi_2=\phi_{(G\cup e \cup f)_\boxempty}$.
We note that if $\Gamma\in \cC''\cup \tilde \cC$ then $\Gamma$ must contain an $e$-pair. Note also that if $\Gamma\in \tilde \cC$ does not contain an $f$-pair, then $\Gamma\in  \cC''$ and $\phi_1(\Gamma)=\phi_2(\Gamma)$. Similarly if $\Gamma\in \cC''$ does not contain an $f$-pair, then $\Gamma\in \tilde \cC$ and $\phi_1(\Gamma)=\phi_2(\Gamma)$. We conclude that for a cluster $\Gamma\in \cC''\cup \tilde \cC$ to contribute to the difference of sums in~\eqref{eqPhi12}$, \Gamma$ must contain both an $e$-pair and an $f$-pair. Thus (dropping the subscripts from $\phi_1, \phi_2$)
\begin{align}\label{eqRDef}
\left|\log \left( \frac{Z_{(G\cup e)_\boxempty }(\lam)}{Z_{G_\boxempty}(\lam)} \right)- \log \left( \frac{Z_{(G\cup e\cup f)_\boxempty}(\lam)}{Z_{(G \cup f)_\boxempty}(\lam)} \right)\right|
\leq R\, .
\end{align}
We note that if $\Gamma$ contains both an $e$-pair and an $f$-pair then then $|\Gamma|\geq 3$ (we may assume that $e\neq f$ else the lemma is trivial). Since there are at most $n$ $e$-pairs an application of Lemma~\ref{lemClusterTail} (applied with $k=3$ and $U$ an $e$-pair) shows that 
\begin{align*}
   \sum_{\substack{\Gamma\in \cC_{e,f}((G\cup 
 e\cup f)_\boxempty)}}|\phi(\Gamma)|\lam^{|\Gamma|}
& =O(n\Delta((G\cup 
 e\cup f)_\boxempty)\lam^3)\\
& = O(n^2q\lam^3)\\
 &= O(C^4e^{-C^2/5}) 
\end{align*}
where for the penultimate inequality we used that $\Delta((G\cup 
 e\cup f)_\boxempty)=O(nq)$ since $G\in\Omega(\lam)$. The same bound also holds holds for $\sum_{\substack{\Gamma\in \cC_{e,f}((G\cup 
 e)_\boxempty)}}|\phi(\Gamma)|\lam^{|\Gamma|}$. It follows that $R=O(C^4e^{-C^2/5})$ and so $R\leq 1$ for $C$ sufficiently large. As a consequence
\begin{align*}
\left|
\frac{\lam Z_{(G\cup e)_\boxempty }(\lam)}{Z_{G_\boxempty}(\lam)}
-\frac{\lam Z_{(G\cup e \cup f)_\boxempty}(\lam)}{Z_{(G \cup f)_\boxempty}(\lam)}
\right|
&\leq 
\frac{\lam Z_{(G\cup e)_\boxempty }(\lam)}{Z_{G_\boxempty}(\lam)}
\left|e^R-1\right|\\
 &\leq
2\lam e^{-\lam^2n/5}R
\end{align*}
where for the second inequality we used~\eqref{eqLkDiffconc} and the fact that $R\leq 1$. 
The result now follows from~\eqref{eqSDbyLocal} and the fact that $|x/(1+x)-y/(1+y)|<|x-y|$ for $x,y>0$.
\end{proof}

\subsection{Bounding cluster sums}
As in the previous section, we fix a weakly balanced partition $(A,B)$.

Our goal in this section is to effectively bound the sums over clusters appearing in Lemma~\ref{lemMargDiff}.

Given a graph $H$, let $\mathcal{Q}(H)$ denote the set of labeled spanning trees of $H$. The following is the tree--graph bound of Penrose which will be convenient when bounding sums of Ursell functions. 
\begin{lemma}[\cite{penrose1967convergence}] \label{lempenrosetree}
For a graph $H=(V,E)$, 
\[
\left| \sum_{\substack{A \subseteq E:\\ (V, A)\textup{ connected}}}  (-1)^{|A|}\right| \leq |\mathcal{Q}(H)|\, .
\]
\end{lemma}

Given a graph $G$ and cluster $\Gamma\in\cC(G)$ such that $|\Gamma|=k$, we identify the vertex set of $G_\Gamma$ (the incompatibility graph of $\Gamma$) with $[k]$.

\begin{lemma}\label{lemHomSum}
Let $k\geq 1$ and let $Q$ be a tree on vertex set $[k]$. Let $G\subseteq\binom{A}{2}\cup \binom{B}{2}$, $f\in \binom{A}{2}\cup \binom{B}{2}$.
\begin{align}\label{eqHomSum}
\sum_{\substack{e\in\binom{A}{2}\cup \binom{B}{2},\\ e\neq f}}\, \sum_ {\substack{\Gamma\in \cC_{e,f}((G\cup 
 e \cup f)_\boxempty):\\ |\Gamma|= k}}
\mathbf 1_{Q\in \mathcal{Q}(H_{e,\Gamma})}
 \leq 2k^5n^2(\Delta+1)^{k-3}\, ,
\end{align}
where $\Delta=\Delta((G \cup f)_\boxempty)$ and $H_{e,\Gamma}$ is the incompatibility graph of $\Gamma$ with respect to the graph $(G\cup 
 e \cup f)_\boxempty$ .
\end{lemma}
\begin{proof}Recall that a cluster $\Gamma$ of size $k$ is a $k$-tuple of elements of $A\times B$, i.e., $\Gamma$ is a map $[k]\to A \times B$.

We pick an $f$-pair $\tilde f=\{u,v\}$ (there are at most $n$ choices for $\tilde f$) and $z,w\in [k]$, $\{x,y\}\in E(Q)$ (there are at most $2k^3$ choices for $(x,y,w,z)$) and construct a cluster $\Gamma$ such that for some $e\in \binom{A}{2}\cup  \binom{B}{2}$ where $e\neq f$:
\begin{enumerate}
\item  $Q\in \mathcal{Q}(H_{e,\Gamma})$.
\item $\Gamma\in \cC_{e,f}((G\cup 
 e \cup f)_\boxempty)$.
 \item $\Gamma(z)=u, \Gamma(w)=v$. 
 \item Either 
 \begin{enumerate}
 \item $\Gamma(a)=\Gamma(b)$ or $\Gamma(\{a,b\})$ is an edge of $(G \cup f)_\boxempty$ for all $\{a,b\}\in E(Q)$ (in other words, $Q$ is a labeled spanning tree of the incompatibility graph of $\Gamma$ with respect to $(G\cup f)_\boxempty$), or
 \item
 $\Gamma(\{x,y\})$ is not an edge of $(G \cup f)_\boxempty$ and $\Gamma(x)\neq\Gamma(y)$ and $\{x,y\}$ is the closest edge of $Q$ to $\{w,z\}$ (in terms of graph distance in $Q$) with these properties.
 \end{enumerate}
\end{enumerate}

Suppose without loss of generality that $\delta_Q(\{x,y\},\{w,z\} )=\delta_Q(x,w)$ where $\delta_Q$ denotes graph distance in $Q$.
Let $\{z',z\}$ be an edge of $Q$ such that $z$ is in a separate component to $x,y,w$ in the graph $Q$ with the edge $\{z,z'\}$ removed. Consider the graph obtained by deleting edges $\{z,z'\}$ and $\{x,y\}$ from $Q$ and call its three components $Q_1, Q_2, Q_3$. Assume without loss of generality that $Q_1$ contains $w$ and $x$, $Q_2$ contains $z$, and $Q_3$ contains $y$. Let $P$ be the path in $Q_1$ from $w$ to $x$. By (4), $\Gamma(a)=\Gamma(b)$ or $\Gamma(\{a,b\})$ is an edge of $(G \cup f)_\boxempty$ for all $\{a,b\}\in E(P)$. Since $\Gamma(w)$ is fixed, there are at most $(\Delta+1)^{|P|}$ choices for $(\Gamma(v): v\in V(P))$. Given $(\Gamma(v): v\in V(P))$ (and so in particular $\Gamma(x)$) there are at most $ n$ choices for $\Gamma(y)$ since $\Gamma(\{x,y\})$ need not be an edge of $(G \cup f)_\boxempty$ but does need to be an edge of $(G \cup e\cup f)_\boxempty$ for some $e$. If $\Gamma(\{x,y\})$ is not an edge of $(G \cup f)_\boxempty$ then $\Gamma(\{x,y\})$ specifies the pair $e$ for which (1) and (2) hold. 

Observe that since $Q\in \mathcal{Q}(H_{e,\Gamma})$ we have that either $\Gamma(a)=\Gamma(b)$ or $\Gamma(\{a,b\})$ is an edge of $(G \cup e\cup f)_\boxempty$ for all $\{a,b\}\in E(Q)$.  
We conclude that there are at most $(\Delta+2)^{|Q_1|-|P|}$ choices for  $(\Gamma(v): v\in V(Q_1)\backslash V(P))$ since  $(G \cup e\cup f)_\boxempty$ has maximum degree at most $\Delta+1$. Similarly there are at most $(\Delta+2)^{|Q_2|}$ choices for $(\Gamma(v): v\in V(Q_2))$, and at most $(\Delta+2)^{|Q_3|}$ choices for $(\Gamma(v): v\in V(Q_3))$ (since $\Gamma(y)$ has been specified).
 In total there are at most
\[
2k^3n^2(\Delta+2)^{|Q_1|+|Q_2|+|Q_3|}= 2k^3n^2(\Delta+2)^{k-3}
\]
choices for $\Gamma$ satisfying (1)-(4) for some $f$-pair $\{u,v\}$, some $\{x,y\}\in E(Q)$ and some $e\in\binom{A}{2}\cup\binom{B}{2}$, $e\neq f$. Finally we note that every such cluster is counted at most $k^2$ times in the sum~\eqref{eqHomSum}.
\end{proof}

 Let $Q_k$ denote the set of all labeled trees on vertex set $[k]$.
\begin{lemma}\label{fixedclusterphibd}
If $G\subseteq\binom{A}{2}\cup \binom{B}{2}$, $f\in \binom{A}{2}\cup \binom{B}{2}$, then
\begin{align}\label{eq:LHStobe}
\sum_{\substack{e\in\binom{A}{2}\cup\binom{B}{2},\\ e\neq f}}\, \sum_ {\substack{\Gamma\in \cC_{e,f}((G\cup e \cup f)_\boxempty):\\ |\Gamma|= k}}
|\phi(\Gamma)|  \leq 20^k n^2 \Delta^{k-3}\, ,
\end{align}
where $\Delta=\Delta((G \cup f)_\boxempty)$.
\end{lemma}

\begin{proof}
By Lemma~\ref{lempenrosetree}, the LHS of~\ref{eq:LHStobe} is at most
\begin{align*}
  &\frac{1}{k!}\sum_{\substack{e\in\binom{A}{2}\cup\binom{B}{2},\\ e\neq f}}\, \sum_ {\substack{\Gamma\in \cC_{e,f}((G\cup 
 e\cup f)_\boxempty):\\ |\Gamma|= k}}
\left|  \sum_{\substack{A \subseteq E(H_{e,\Gamma}):\\ ([k],A)\textup{ connected}}}  (-1)^{|A|}
  \right|\\
   &\leq
  \frac{1}{k!} \sum_{\substack{e\in\binom{A}{2}\cup\binom{B}{2},\\ e\neq f}}\, \sum_ {\substack{\Gamma\in \cC_{e,f}((G\cup 
 e\cup f)_\boxempty):\\ |\Gamma|= k}}
\sum_{Q\in Q_k}\mathbf 1_{Q\in \mathcal{Q}(H_{e,\Gamma})} \\
 &\leq
 \frac{1}{k!}\sum_{Q\in Q_k} \sum_{\substack{e\in\binom{A}{2}\cup\binom{B}{2},\\ e\neq f}}\, \sum_ {\substack{\Gamma\in \cC_{e,f}((G\cup 
 e\cup f )_\boxempty):\\ |\Gamma|= k}}
\mathbf 1_{Q\in \mathcal{Q}(H_{e,\Gamma})}\\
 &\leq 20^k n^2 \Delta^{k-3}
 \end{align*}
 where for the final inequality we used Lemma~\ref{lemHomSum} and Cayley's formula: $|Q_k|=k^{k-2}$\, .
\end{proof}

In the next subsection, we prove Lemma~\ref{lempathcoupleNu} and hence also  Proposition~\ref{propNuMix}.

\subsection{Proof of Lemma~\ref{lempathcoupleNu}}
As before, it suffices to consider the case where $d(X_t, Y_t)=1$ i.e. $X_t$ and $Y_t$ differ in one edge.
Suppose WLOG that $f$ is an edge of $X_t$, but not $Y_t$. Consider a randomly chosen $e \in \binom{A}{2}\cup\binom{B}{2}$ chosen for update. There are three cases to consider:

{\em Case 1: $e=f$.}
In this case
\[
\E[d(X_{t+1},Y_{t+1}) - d(X_t,Y_t) \mid e]=-1\, .
\]

{\em Case 2: $X_t \cup e$ contains a triangle:} In this case
\[
\E[d(X_{t+1},Y_{t+1}) - d(X_t,Y_t) \mid e] \leq q
\]
since the distance increases only if $Y_{t+1}=Y_t\cup e$ which happens with probability at most $q$ by Lemma~\ref{lemSDedgemarg}.
We note that the number of pairs $e$ such that $X_t \cup \{e,f\}$ contains a triangle is at most $2qn$ (it is here that we use the assumption $X_t\in\Omega(\lam)$).

{\em Case 3: $X_t \cup e$ does not contain a triangle and $e\neq f$.} Given such an $e$ we have
\[
\E[d(X_{t+1},Y_{t+1}) - d(X_t,Y_t)\mid e] \leq \left|p(e | X_t) - p(e | Y_t)\right|
\]
since the only way for the distance to increase is if $X_{t+1}=X_t\cup e$ and $Y_{t+1}=Y_t$ which occurs with probability at most $\left|p(e | X_t) - p(e | Y_t)\right|$ by the definition of the optimal coupling. Here we have used that $X_t\in \Omega(\lam)$ so that $\Delta(X_t)\leq 2n\lam e^{-n\lam^2/5}\leq \alpha/\lam -1$ and so $X_t \cup e$ does not violate the maximum degree condition in the definition of $\cD^{\textup{w}}_{A,B,\lam}$. Let $H$ denote the set of all $e\neq f$ such that $X_t \cup e$ does not contain a triangle. Let $q:=\lam e^{-\lam^2n/5}$ and for $e\in H$ let
\[
R_e:=
\sum_{\substack{\Gamma\in \cC_{e,f}((Y_t\cup e)_\boxempty)}}|\phi(\Gamma)|\lam^{|\Gamma|}
 +
 \sum_{\substack{\Gamma\in \cC_{e,f}((X_t\cup e)_\boxempty)}}|\phi(\Gamma)|\lam^{|\Gamma|}\, .
\]
By Lemmas~\ref{lemMargDiff} and~\ref{fixedclusterphibd} (recalling that $X_t=Y_t\cup f$) we conclude that 
\begin{align*}
&\sum_{e\in H}\E[d(X_{t+1},Y_{t+1}) - d(X_t,Y_t)\mid e]\\
&\leq
2q \sum_{e\in H} R_e\\
 &\leq 
 4qn^2\sum_{k\geq 3}20^k (2nq)^{k-3}\lam^k\\
 &\leq 10^5 qn^2\lam^3
\end{align*}
where for the final inequality we used that $40\lam nq\leq 1/2$.
Letting $N=\binom{|A|}{2}+\binom{|B|}{2}$, it follows that,
\begin{align*}
    \E[d(X_{t+1},Y_{t+1})]-1
    &\leq \frac{1}{N}\left(-1 +2q^2n+ 10^5 qn^2\lam^3\right)\\
    &\leq -\frac{1}{2N}\\
    &\leq -\frac{1}{n^2}\, ,
\end{align*}
for $C$ sufficiently large.
\qed

\subsection{Sampling from the defect distribution}
With Proposition~\ref{propNuMix} in hand, we now prove Proposition~\ref{propDefectSample} and show that we can approximately sample from the defect distributions $\nu_{A,B,\lam}$. We will show that we can implement each step of the Glauber dynamics efficiently and up to small error. For this, we need the following lemma that will allow us to approximate edge marginals in $\nu_{A,B,\lam}$.

\begin{lemma}\label{lemHCRandApprox}[\cite{chen2022localization},~\cite{jerrum1986random}]
Let $G$ be a graph on $n$ vertices and let $\lam\leq 1/\Delta(G)$. Given $\eps, \delta>0$ there exist randomized algorithms that runs in time polynomial in $n, \log(1/\eps), \log(1/\delta)$ that 
 \begin{enumerate}
        \item Output $M$  so that with probability at least $1-\delta$, $M$ is an $\eps$-relative approximation to $Z_G(\lam)$.
        \item Output $G$ with distribution $\hat\mu$ so that $\| \hat \mu - \mu \|_{TV} \le \eps$ where $\mu$ is the hard-core measure on $G$ at activity $\lam$\,.
    \end{enumerate}
\end{lemma}

\begin{proof}[Proof of Proposition~\ref{propDefectSample}]
Part (1) of the Proposition~\ref{propDefectSample} follows from part (2) via a standard reduction of approximate counting to sampling. We now prove part (2). Let $(X_t)_{t\geq 0}$ be a run of the Glauber dynamics for $\nu_{A,B,\lam}$ with $X_0$ arbitrary. By Proposition~\ref{propNuMix} there exists $K>0$ such that if $T:=Kn^2\log(n/\eps)$ then $\|\cL(X_T)-\nu_{A,B,\lam}\|_{TV}\leq \eps$. We will define a dynamic $(Y_t)$ that will serve as an approximation to $(X_t)$. 

Recall from $\eqref{eqSDbyLocal}$ that if $\nu_{A,B,\lam}(G\cup e)>0$ then
\begin{align}
p(e | G)
= \frac{\lam Z_{(G\cup e)_\boxempty }(\lam)}{Z_{G_{\boxempty}}(\lam)+\lam Z_{(G\cup e)_\boxempty }(\lam)}\, ,
\end{align}
and $p(e | G)=0$ otherwise. Let $\delta=1/T$ and let $\eps'=\eps/T$. By Lemma~\ref{lemHCRandApprox}, in $\text{poly}(n, 1/\eps)$ time we can compute $M_1, M_2$ in  so that with probability $\geq 1-\delta$, $M_1, M_2$ are $\eps'$-relative approximations to $Z_{(G\cup e)_\boxempty }(\lam), Z_{G_\boxempty}(\lam)$ respectively. Let
\[
\hat p(e|G) = \frac{\lam M_1}{M_2+\lam M_1}\, ,
\]
and note that with probability $\geq 1-\delta$, $\hat p(e | G)$ is a $2\eps'$-relative approximation to $p(e | G)$ and so $|p(e | G)-\hat p(e | G)|\leq e^{2\eps'}-1\leq 3\eps'$ since $p(e | G)\leq 1$.

Define $(Y_t)$ as follows. 
Given $Y_t$ perform the following update:
\begin{itemize}
\item Pick $e\in\binom{A}{2}\cup \binom{B}{2}$ uniformly at random and let $U\sim U[0,1]$.
\item Compute $\hat p(e|Y_t)$. If $U\leq \hat p(e | Y_t)$ then set $Y_{t+1}=Y_t\cup e$. Otherwise set $Y_{t+1}=Y_t\backslash e$. 
\end{itemize}
We note that we can compute $Y_T$ in $\text{poly}(n, 1/\eps)$ time and by coupling $(Y_t), (X_t)$ in the obvious way we see that
\[
\P(X_T= Y_T)\geq [(1-\delta)(1-3\eps)]^T\geq 1-T(\delta + 3\eps')\geq 1-4\eps\, .
\]
By the coupling inequality we conclude that
\[
\|\cL(Y_T)- \cL(X_T)\|_{TV} \leq \P(X_T\neq Y_T)\leq 4\eps\, ,
\]
and so by the triangle inequality
\[
\|\cL(Y_T)- \nu_{A,B,\lam}\|_{TV}\leq 5\eps.\qedhere
\]
\end{proof}

\subsection{Proof of Theorem~\ref{thmHighSampling}}

First we prove part (1) of the theorem. First observe that by symmetry, $Z_{A,B}^{\textup{w}}(\lam)$ depends only on the level of imbalance $||A|-|B||$ of the partition $(A,B)$. If $k=||A|-|B||$ for some weakly balanced partition $(A,B)$ we set  $Z_k(\lam):=Z_{A,B}^{\textup{w}}(\lam)$ (so in particular $k\leq n/10$). By Proposition~\ref{propDefectSample}, for each $k\in [0,n/10]$ we can compute $M_k$ in $\text{poly}(n,1/\eps)$ time such that with probability $\geq 1-1/n$, $M_k$ is an $\eps$-relative approximation to $Z_k(\lam)$. We can therefore compute
\[
M:=\sum_{k=0}^{n/10}\binom{n}{n/2 + k/2}M_k
\]
in $\text{poly}(n,1/\eps)$ time and, by a union bound, with probability $\geq 9/10$, $M$ is an  $\eps$-relative approximation to 
\[
Z_{\textup{weak}}(\lam)= \sum_{k=0}^{n/10}\binom{n}{n/2 + k/2}Z_k(\lam)\, .
\]
Part (1) of the theorem now follows from Proposition~\ref{PropError}.

We now turn to part (2). 
Let $\eps>0$. By Proposition~\ref{PropError} it suffices to show that we can sample from $\mu_{\mathrm{weak},\lam}$ up to total variation distance $\eps$ in time polynomial in $n$ and $1/\eps$. For this, we show that we can implement each step in the definition of $\mu_{\mathrm{weak},\lam}$ up to total variation distance $\eps$ in time polynomial in $n$ and $1/\eps$. 

\noindent{\emph{Sampling a partition.}} Step~\ref{S1} in the definition of $\mu_{\mathrm{weak},\lam}$ can be rewritten as:
\begin{enumerate}
\item Sample an integer $k$ with probability proportional to $Z_k(\lam)$.
\item Sample a partition $(A,B)$ of $[n]$ such that $||A|-|B||=k$ uniformly at random. 
\end{enumerate}
As above it suffices to observe that for each $k\in [0,n/10]$ we can compute $M_k$ in $\text{poly}(n,1/\eps)$ time such that with probability $\geq 1-1/n$, $M_k$ is an $\eps$-relative approximation to $Z_k(\lam)$.\\

\noindent{\emph{Sampling defects.}} Proposition~\ref{propDefectSample} shows that we can implement Step~\ref{S2} in the definition of $\mu_{\mathrm{weak},\lam}$, i.e., we can sample from $\nu_{A,B,\lam}$ up to $\eps$ total variation distance in $\text{poly}(n, 1/\eps)$ time. \\

\noindent{\emph{Sampling crossing edges.}} Given $(A,B)$ and $(S,T)\in \cD_{A,B,\lam}^{\textup{w}}$ from Steps~\ref{S1} and~\ref{S2}, we note that $\lam \leq 1/\Delta(S\boxempty T)$ since $\Delta(S \boxempty T)\leq 2\Delta(S \cup T)\leq 2\alpha/\lam$ (by the definition of $\cD_{A,B,\lam}^{\textup{w}}$). We may therefore apply Lemma~\ref{lemHCRandApprox} to output sample from the hard-core model on $S\boxempty T$ at activity $\lam$ up to total variation distance $\eps$ in $\text{poly}(n, 1/\eps)$ time.

\section{Proof of Proposition~\ref{PropError}}\label{secMuWeakError}

We begin with a special case of~\cite[Theorem 1.7]{balogh2015independent}.

\begin{theorem}\label{thmluczakerror}
For all $\delta>0$ there exist $C,c>0$ so that if $m \geq Cn^{3/2}$, then all but a $e^{-cm}$ proportion of graphs in $\cT(n,m)$ admit a cut of size at least $(1-\delta) m$.
\end{theorem}

We note that in~\cite{balogh2015independent} the result is stated with $o(1)$ in place of the effective error $e^{-cm}$. However, as  noted by in previous works (see e.g. the comment following~\cite[Theorem 1.2]{balogh2016typical}) the proof in~\cite{balogh2015independent} gives the error stated above.  We remark that Theorem~\ref{thmluczakerror} was first proved (with $o(1)$ in place of $e^{-cm}$) by T.\ \L uczak~\cite{luczak2000triangle}.

We will in fact need a slight refinement of  Theorem~\ref{thmluczakerror} (Theorem~\ref{thmluczakrefined} below) which follows from combining~\cite[Proposition 6.1]{balogh2016typical} and~\cite[Claim 6.2]{balogh2016typical}. To state the result we need a  definition.

\begin{defn}\label{def:dominating}
Let $G$ be a graph and $(A,B)$ a partition of its vertex set. We say $(A,B)$ is a \emph{dominating} cut of $G$ if
\[
d_G(v,B)\geq d_G(v,A) \text{ for all } v\in A
\]
and similarly with $A,B$ swapped.
\end{defn}
Recall also the definition of a weakly balanced partition from Definition~\ref{defWeakBalance}.

\begin{theorem}\label{thmluczakrefined}
For all $\delta>0$ sufficiently small, there exist $C,c>0$ so that if $m \geq Cn^{3/2}$, then all but a $e^{-cm}$ proportion of graphs in $\cT(n,m)$ admit a dominating, weakly balanced cut of size at least $(1-\delta) m$.
\end{theorem}

The results of~\cite{balogh2016typical} are also stated with a $o(1)$ error in place of $e^{-cm}$. However one readily checks that the short proof of Theorem~\ref{thmluczakrefined}
from Theorem~\ref{thmluczakerror} presented in~\cite{balogh2016typical} yields the error stated here. 

For $\lam, \delta>0$, let $\mathcal{L}(\lam,\delta)$ denote the set of all $G\in \cT$ such that $G$ admits a weakly balanced, dominating cut of size $\geq |G| - 2\delta \lambda n^2$.

For a subset $\mathcal{R}\subseteq \cT$ and $\lam>0$ we define the measure on $\mathcal{R}$ given by
\[
\mu_{\mathcal{R},\lam}(G) =\frac{\lam^{|G|}}{Z(\mathcal{R},\lam)}\, ,
\]
where 
\[
Z(\mathcal{R},\lam)= \sum_{G\in \mathcal{R}}\lam^{|G|}\, .
\]
We use the following result from~\cite{jenssen2023evolution} (Proposition 3.4). Recall that we let $Z(\lam)=Z(\mathcal{T},\lam)$.

\begin{prop}\label{lemZerothApprox}
For all $\delta>0$ sufficiently small, there exists $C,c>0$ such that
if $\lam\geq C/\sqrt{n}$ then
\[
\|\mu_{\cT,p} - \mu_{\cL,\lam}\|_{TV}\leq e^{-c\lam n^2}\, ,
\]
and
\[
Z(\lam)=(1+O(e^{-c\lam n^2}))Z(\mathcal{L},\lam)\, ,
\]
where $\cL=\cL(\lam,\delta)$ and $p=\lam/(1+\lam)$ 
\end{prop}
The proposition is an easy consequence of Theorem~\ref{thmluczakrefined}. Here we have included the effective error that the proof in~\cite{jenssen2023evolution} gives. 

Finally we note the following result from~\cite{jenssen2023evolution} (Proposition 3.5).

\begin{prop}
\label{propGroundStateRefinedAlt}
For all $\delta>0$ sufficiently small, there exists $C>0$ such that
if $\lam\geq C/\sqrt{n}$ then
\[
\|\mu_{\cL,\lam} - \mu_{\textup{weak},\lam}\|_{\textup{TV}}= O(e^{-\sqrt{n}})\, ,
\]
and 
\[
Z(\cL, \lam) = \left(1+O(e^{-\sqrt{n}}) \right) Z_{\textup{weak}}(\lam)\, ,
\]
where $\cL=\cL(\lam,\delta)$.
\end{prop} 
Combining the above two propositions, we immediately arrive at a proof of Proposition~\ref{PropError}.

\section{Slow mixing at high density}\label{secSlowMix}

In this section we prove Theorem~\ref{thmSlowMix}. We import some structural results from~\cite{jenssen2023evolution}. Roughly speaking, these results will show that almost all graphs in $\cT(n)$ have a unique max cut in a rather robust sense. 

First we need a definition. Given a graph $G$ and subsets $X,Y\subseteq V(G)$, we let $E(X,Y)$ denote the set of edges with one endpoint in $X$ and one endpoint in $Y$.
\begin{defn}\label{defABExpand}
 Given a graph $G$ and a partition $(A,B)$, we call $G$ an $(A,B)$-$\lam$-expander if 
\[
\text{$d_G(v, B)\geq \lam n/30$ for all $v\in A$}\, ,
\] 
and 
\[
|E(X,Y)|\geq \lam |X||Y|/10
\]
whenever $X\subseteq A,\,  |X|\geq \lam n/100$, and  $Y\subseteq B,\,  |Y|\geq n/6$,
and both statements hold also with $A,B$ swapped.
 \end{defn}

 We use Lemma 5.3 from~\cite{jenssen2023evolution}. 
 \begin{lemma}\label{lemExpanderwhp}
There exists $C>0$ so that if $\lam\geq C/\sqrt{n}$ the following holds. 
Let $(A,B)$ be a weakly balanced partition, and let $(S,T)\in \cD_{A,B,\lam}^{\textup{w}}$. Sample $\Ec\subseteq A\times B$ according to the hard-core model on $S\boxempty T$ at activity $\lam$. Let $G$ be the graph $([n], \Ec)$. Then
\[
\P(G \textup{ is an $(A,B)$-$\lam$-expander} )\geq 1-  e^{-\sqrt{n}}\, .
\]
\end{lemma}

For a weakly balanced partition $(A,B)$, let $\cE_{A,B,\lam}$ denote the set of $G\in \cT(n)$ such that $(G[A], G[B])\in \cD_{A,B,\lam}^{\textup{w}}$ and $G$ is an $(A,B)$-$\lam$-expander.
Let 
\[
\cE_0:= \cT(n)\backslash \bigcup_{\substack{(A,B)\\ \text{ weakly balanced}}}\cE_{A,B,\lam}\, .
\]

Let $p=\lam/(1+\lam)$. We will use a conductance argument using $\cE_0$ as a `bottleneck' to lower bound the mixing time of the Glauber dynamics of $\mu_{\cT,p}$.

If $G\in \cE_0$ then either
\begin{enumerate}
\item $(G[A], G[B])\notin \cD_{A,B,\lam}^{\textup{w}}$ for all weakly balanced $(A,B)$, or
\item $(G[A], G[B])\in \cD_{A,B,\lam}^{\textup{w}}$ for some weakly balanced $(A,B)$ but $G$ is not an $(A,B)$-$\lam$-expander.
\end{enumerate}
Let $\cE_0^1, \cE_0^2$ denote the set of $G\in \cE_0$ satisfying (1), (2) respectively. Proposition~\ref{PropError} shows that
\(
\mu_{\cT,p}(\cE_0^1)= O(e^{-\sqrt{n}})
\)
and Proposition~\ref{PropError} taken together with Lemma~\ref{lemExpanderwhp} shows that
\(
\mu_{\cT,p}(\cE_0^2)= O(e^{-\sqrt{n}}).
\)
We conclude that
\[
\mu_{\cT,p}(\cE_0)=O(e^{-\sqrt{n}})\, .
\]

Note that by symmetry, if $(A,B), (A', B')$ are two partitions such that $||A|-|B||=||A'|-|B'||$ then $\mu_{\cT,p}(\cE_{A,B,\lam})=\mu_{\cT,p}(\cE_{A',B',\lam})$. Moreover, if $(A,B)$ is weakly balanced then there are at least $\binom{n}{n/2-n/10}=e^{\Omega(n)}$ other partitions with the same imbalance. It follows that 
\[
\mu_{\cT,p}(\cE_{A,B,\lam})\leq e^{-\Omega(n)}\, 
\]
for all weakly balanced $(A,B)$. It follows that we can take a union of sets $\cE_{A,B,\lam}$, call it $\cE_1$, so that 
\[
\frac{1}{4}\leq \mu_{\cT,p}(\cE_1)\leq \frac{1}{2}\, .
\]
We will show that for the Glauber dynamics to leave the set $\cE_1$ it must pass through $\cE_0$. By a standard conductance argument (see \cite[Claim 2.3]{dyer2002counting}) this implies that the mixing time is at least
\[
\frac{\mu_{\cT,p}(\cE_1)}{8\mu_{\cT,p}(\cE_0)}\geq \Omega (e^{\sqrt{n}})\, .
\]

It remains to show that for the Glauber dynamics to leave the set $\cE_1$ it must pass through $\cE_0$. Since the Glauber dynamics updates a single edge at a time, it suffices to show that if $G_1\in \cE_{A,B,\lam}$ and $G_2\in \cE_{A',B',\lam}$ for distinct weakly balanced partitions $(A,B), (A',B')$ then $|E(G_1)\Delta E(G_2)|\geq 2$. This is what we show now. 

Since the partitions $(A,B), (A',B')$ are distinct, either $A \cap B'\neq \emptyset$ or $B\cap A'\neq \emptyset$. 
Assume WLOG that $B\cap A'\neq \emptyset$.

Suppose first that $|A\cap B'|<\lam n /100$.
By assumption there exists $v\in B\cap A'$. Since $G_1$ is an $(A,B)$-$\lam$-expander we then have 
\begin{align*}
    d_{G_1}(v, A')
    &\geq d_{G_1}(v, A) - d_{G_1}(v, A\cap B')\\
    &\geq \lam n/30 - |A\cap B'|\\
    &\geq \lam n /50\, .
\end{align*}
On the other hand, since $G_2\in \cD_{A',B',\lam}^{\textup{w}}$, $d_{G_2}(v, A')\leq \alpha/\lam$ and so 
\[
|E(G_1)\Delta E(G_2)|\geq \lam n /50 - \alpha/\lam \geq 2.
\]
We may therefore assume that $|A\cap B'|\geq \lam n /100$. 
In particular $A\cap B'\neq\emptyset$, and so by a similar argument to the above, we may assume that $|A'\cap B|\geq \lam n /100$ also. 
By symmetry (swapping the roles of $A$ and $B$) we may also assume that $|A \cap A'|\geq \lam n /100$ and $|B\cap B'| \geq \lam n/100$.

Since $(A,B)$ is weakly balanced we have $|B| \geq n/3$. Suppose WLOG that $|B\cap B'|\geq |B \cap A'|$ so that in particular $|B\cap B'|\geq n/6$.
Since $G_1$ is an $(A,B)$-$\lam$-expander we then have
\begin{align}\label{eqABlamExpand}
|E_{G_1}(A\cap B',B\cap B')|\geq \lam |A\cap B'||B\cap B'|/10\, ,
\end{align}
and so there exists $v\in A\cap B'$ such that $d_{G_1}(v, B\cap B')\geq \lam |B\cap B'|/10\geq \lam n /60$. On the other hand, since $G_2\in \cD_{A',B',\lam}^{\textup{w}}$, $d_{G_2}(v, B')\leq \alpha/\lam$ and so 
\[
|E(G_1)\Delta E(G_2)|\geq \lam n /60 - \alpha/\lam \geq 2.
\] 
This concludes the proof.  \qedsymbol

\section*{Acknowledgments}

MJ is supported by a UK Research and Innovation Future Leaders Fellowship MR/W007320/2.  WP supported in part by NSF grant CCF-2309708. AP is supported by an NSERC Discovery grant. MS supported in part by NSF grant DMS-2349024.

\bibliography{triangle.bib}
\bibliographystyle{abbrv}

\end{document}